\documentclass[letterpaper, 10 pt, conference]{ieeeconf}

\IEEEoverridecommandlockouts
\usepackage{cite}
\usepackage{amsmath,amssymb,amsfonts}
\usepackage{algorithmic}
\usepackage{graphicx}
\usepackage{textcomp}
\usepackage{xcolor}
\usepackage{url}

\makeatletter
\@ifundefined{proof}{}{}
\@ifundefined{endproof}{}{}
\makeatother
\usepackage{amsthm}
\usepackage{overpic}
\usepackage{lipsum}
\makeatletter
\@ifundefined{labelindent}{}{\let\labelindent\relax}
\makeatother
\usepackage{enumitem}
\usepackage{cases}
\usepackage{subcaption}
\usepackage{graphicx}
\usepackage{ulem}
\usepackage{xcolor}
\usepackage{cite}
\usepackage{cancel}

\usepackage{algorithm,algorithmic}
\usepackage{mathbbol}
\usepackage{comment}

{
    \theoremstyle{plain}
    
}

\DeclareCaptionLabelFormat{lc}{\MakeLowercase{#1}~#2}

\newtheorem{proposition}{Proposition}

\def\BibTeX{{\rm B\kern-.05em{\sc i\kern-.025em b}\kern-.08em
    T\kern-.1667em\lower.7ex\hbox{E}\kern-.125emX}}

\renewcommand{\baselinestretch}{0.964}

\title{\LARGE \bf A Mean-Field Approach for Safe Routing of Multi-Destination Urban Air Mobility Networks
}

\author{Nameer Fawwaz Ahmed, 
Cody Fleming, and
Yasser Shoukry
\thanks{ 
This work was supported in part by  NSF Grant 2313104.
}
\thanks{
Nameer Fawwaz Ahmed and Yasser Shoukry are with the Department of Electrical Engineering and Computer Science at the University of California, Irvine, CA 92697, USA (e-mail: ahmednf@uci.edu, yshoukry@uci.edu).}
\thanks{Cody Fleming is with the Mechanical Engineering
Department, Iowa State University, IA 50011, USA (e-mail: flemingc@iastate.edu).}
}

\begin{document}

\maketitle

\begin{abstract}
As Urban Air Mobility (UAM) systems scale toward high-density operations, managing autonomous Unmanned Aerial Vehicle (UAV) traffic requires control frameworks that are both tractable and safety-critical. This paper presents a principled optimal control-theoretic foundation for routing in multi-destination UAM networks subject to vertiport capacity and flow constraints. We first model the network as a destination-conditioned Continuous-Time Markov Chain (CTMC) to capture the stochastic transitions between queueing, service, and flight states. To ensure tractability, we employ a mean-field fluid approximation and derive the underlying system dynamics as a set of coupled ordinary differential equations. A key contribution of this work is the formal proof of the positive invariance of the queue-free state space. We demonstrate that under specific underloaded conditions, a system initialized without queues will remain queue-free indefinitely. This result allows us to transform a complex, infinite-dimensional continuous-time optimal control problem into a tractable, finite-dimensional algebraic optimization. The resulting framework jointly optimizes for travel time and multi-hop efficiency while ensuring network-wide stability. We validate the approach by characterizing the steady-state flow equilibria and providing sufficient conditions for safe, congestion-free operation in large-scale mobility systems.

\end{abstract}



\section{Introduction}
\label{sec:Introduction}
The rapid advancement of electric vertical take-off and landing (eVTOL) technology has positioned Urban Air Mobility (UAM) as a promising solution to urban congestion, but realizing high-throughput, autonomous UAM service requires a control-theoretic framework that can manage dense, stochastic traffic flows within constrained vertiport infrastructure. Prior work has approached this problem through decentralized air traffic management and discrete path-finding formulations such as the Vehicle Routing Problem~\cite{bharadwaj2021decentralized,thibbotuwawa2020unmanned}, and through queueing-theoretic vertiport capacity models~\cite{zhang2022method,han2026rolling}. These approaches, however, either rely on deterministic flight times and static graphs that do not scale to high-density operations, or lack a unified, destination-conditioned routing policy that simultaneously captures microscopic stochastic queueing and provable macroscopic stability. Mean-field techniques have addressed analogous scalability challenges in other mobility sectors~\cite{zheng2021optimal}, but have not yet been extended to provide safety-critical routing guarantees for UAM networks.

This paper closes that gap by developing a mean-field fluid framework for destination-conditioned UAM routing that replaces per-agent stochastic dynamics with a population-level, provably safe description of the network. We model the network as a destination-conditioned Continuous-Time Markov Chain and derive its fluid limit as a set of coupled ordinary differential equations, reducing the state-space dimension from combinatorial in fleet size to linear in the number of vertiports. Our central theoretical contribution is a formal proof that the queue-free state space is positively invariant: under mild underloaded conditions, a network initialized without queues remains queue-free for all time. This result lets us replace the original infinite-dimensional optimal control problem with a tractable, finite-dimensional algebraic optimization over steady-state flow equilibria whose complexity depends only on the number of vertiports rather than the number of UAVs. We validate the resulting framework numerically, characterizing the trade-off between travel efficiency and destination reachability and identifying feasibility and saturation thresholds on vertiport capacity that support principled infrastructure sizing.


\section{Problem Formulation}
\label{sec:Problem Formulation}

\subsection{Notation}

The sets of non-negative integers (including zero) and real numbers are denoted by $\mathbb{N}$ and $\mathbb{R}$, respectively. For a set $\mathcal{S}$, its cardinality is denoted by $|\mathcal{S}|$. A network is defined by a directed graph $\mathcal{G} = (\mathcal{V}, \mathcal{E})$, where $\mathcal{V}$ is the set of vertices (vertiports) and $\mathcal{E}$ is the set of edges (flight links). Matrices are denoted by uppercase letters and their elements by corresponding lowercase letters; for example, $p_{ij}^{(d)}$ represents the $(i,j)$-th entry of the destination-conditioned routing matrix $P^{(d)}$.

The operators $\mathbb{P} (\cdot)$ and $\mathbb{E}[\cdot]$ denote probability and expectation, respectively. The symbol $\mathbf{1}_{[\cdot]}$ represents the indicator function, evaluating to 1 if the condition is true and 0 otherwise (e.g., $\mathbf{1}_{[j \neq d]}$). Conversely, $\mathbf{1}$ denotes a column vector of all ones of appropriate dimension, such that the inner product $\mathbf{1}^\top x$ represents the sum of the components of a column vector $x$.

\subsection{Continuous-Time Markov Chain Modeling of UAM Networks}
We define the UAM network as a directed graph $\mathcal{G}=(\mathcal{V}, \mathcal{E})$ consisting of $V$ vertiports and $E$ links \cite{vascik2019aiaa}. To capture the stochastic nature of the system, we model the network as a \textit{Continuous-Time Markov Chain} (CTMC) where each UAV $n \in \{1, \dots, N\}$ is in one of three states: waiting in a queue, in service on a landing pad, or in flight \cite{vascik2019aiaa}.

Let $d \in \mathcal{V}$ denote the destination of a vehicle. The state of the network at time $t$ is represented by the integer-valued random vector $X(t) = (X^{(1)}(t), \ldots, X^{V}(t))$ such that $X^{(d)}(t)= (Q_i^{(d)}(t), S_i^{(d)}(t), F_{ij}^{(d)}(t))$, where:
\begin{itemize}
    \item $Q_i^{(d)}(t) \in \mathbb{N}$: Number of UAVs  waiting for a pad to be available at vertiport $i$ to be serviced in-route toward destination $d$.
    \item $S_i^{(d)}(t) \in \{ 0, 1, ..., c \} $: Number of UAVs in service on pads at vertiport $i$ with destination $d$.
    \item $F_{ij}^{(d)}(t) \in \mathbb{N} $: Number of UAVs in flight from $i$ to $j$ with final destination set to $d$.
\end{itemize}
The total fleet size $N \in \mathbb{N} $ remains fixed such that $\sum_{d,i} Q_i^{(d)} + \sum_{d,i} S_i^{(d)} + \sum_{d,i,j} F_{ij}^{(d)} = N$.

To simplify the description of capacity-constrained dynamics, we denote the \textbf{aggregate state variables} at vertiport $i$ as:
\begin{align}
    Q_i(t) = \sum_{d \in \mathcal{V}} Q_i^{(d)}(t), \\
    \quad S_i(t) = \sum_{d \in \mathcal{V}} S_i^{(d)}(t), \\
    F_{ij} (t) = \sum_{d \in V} F_{ij}^{(d)} (t)
\end{align}

\paragraph{Service Duration}

Let $\mu_s$ denote the mean service rate per pad. The actual time a UAV spends at a service pad is modeled as a random variable $\tau_s \sim \text{Exp}(\mu_s)$ with mean $\mathbb{E}[\tau_s] = 1/\mu_s$. The probability density function (PDF) is defined as:

\begin{equation}
    \mathbb{P} (\tau_s = t; \mu_s) = \mu_s e^{-\mu_s t}, 
    \qquad
    t \geq 0
    \label{eq:service_pdf}
\end{equation}

\paragraph{Flight Duration}

Let $\mu_{ij}$ denote the mean flight rate on the link $i \to j$ where $i \neq j$. The time spent in flight, $\tau_{ij}$, follows an exponential distribution with rate $\mu_{ij}$, such that $\tau_{ij} \sim \text{Exp}(\mu_{ij})$ and $\mathbb{E}[\tau_{ij}] = 1/\mu_{ij}$. The probability density function (PDF) is defined as:

\begin{equation}
    \mathbb{P} (\tau_{ij} = t; \mu_{ij}) = \mu_{ij} e^{-\mu_{ij} t},
    \qquad
    i \neq j,
    \qquad
    t \geq 0
    \label{eq:flight_pdf}
\end{equation}


Due to the memorylessness of the exponential service and flight time distributions, the random vector process $\{X(t)\}_{t \ge 0}$ constitutes a continuous-time Markov chain (CTMC) \cite{ross2019probability}. A defining characteristic of a CTMC is the Markov property, which dictates that the conditional probability of the future state depends only on the present state and is independent of the past history. Formally, for all times $s, t \ge 0$ and any valid network states $x, y, x(u) \in S$ for $0 \le u < s$:
\begin{align}
& \mathbb{P} (X(s+t) = y \mid X(s) = x, X(u) = x(u) \text{ for } 0 \le u < s) \nonumber \\
&\quad = \mathbb{P} (X(s+t) = y \mid X(s) = x) = \mathbb{P}_{x,y}(t)
\end{align}
This property entails the one does not have to store a history of how long each UAV spends in service or in flight.

However, the state space $S$ grows combinatorially as $|S|=\binom{N+K-1}{K-1}$, where $K = V^2 + V$, making exact computation via Kolmogorov Forward Equations \cite{kolmogoroff1931ueber} computationally intractable for large-scale systems \cite{ethier2009markov, harchol2013performance}.

\subsection{Mean-Field Fluid Approximation of UAM Networks}
To ensure computational tractability as the fleet size $N$ grows, we employ a first-order mean-field fluid approximation. We define the deterministic fluid state variables as the expectations of the underlying stochastic counts: $q_i^{(d)}(t) \approx \mathbb{E}[Q_i^{(d)}(t)]$, $s_i^{(d)}(t) \approx \mathbb{E}[S_i^{(d)}(t)]$, and $f_{ij}^{(d)}(t) \approx \mathbb{E}[F_{ij}^{(d)}(t)]$. We use the normalized counts. We have $ q_i^{(d)}, f_{ij}^{(d)} \in [0,1] $, and $ s_i^{(d)} \in [0, \kappa_i] $ where $ \kappa_i > 0 $.

As $N \to \infty$, the trajectory of the system converges to these expected values according to the {\it{Functional Law of Large Numbers}} (FLLN)~\cite{kurtz1970solutions}. The expectations are taken over the state space $\mathcal{S}$ of the Continuous-Time Markov Chain (CTMC):

\begin{align}
&\mathbb{E}[Q_i^{(d)}(t)] = \sum_{x \in S} Q_i^{(d)}(t) \mathbb{P} (X^{(d)}(t) = x^{(d)}) \\
&\mathbb{E}[S_i^{(d)}(t)] = \sum_{x \in S} S_i^{(d)}(t) \mathbb{P} (X^{(d)}(t) = x^{(d)}) \\
&\mathbb{E}[F_{ij}^{(d)}(t)] = \sum_{x \in S} F_{ij}^{(d)}(t) \mathbb{P} (X^{(d)}(t) = x^{(d)})
\end{align}

In the fluid regime, the discrete stochastic state $X(t)$ is approximated by a continuous deterministic trajectory $x(t) \approx \mathbb{E}[X(t)]$, representing the component-wise expectation of the network's variables across all destinations $d \in \mathcal{V}$.

The destination-conditioned dynamics for the system are governed by the following set of ordinary differential equations (ODEs)~\cite{weiss2021scheduling}:
\begin{equation}
\begin{cases}
\dot{q}_j^{(d)} = a_j^{(d)} - r_j^{(d)} \\
\dot{s}_j^{(d)} = r_j^{(d)} - \mu_s s_j^{(d)} \\
\dot{f}_{ij}^{(d)} = p_{ij}^{(d)} \mu_s s_i^{(d)} - \mu_{ij} f_{ij}^{(d)}
\label{eq: Dynamical System}
\end{cases}
\end{equation}
where the components are defined as follows:
\begin{itemize}
    \item \textbf{Arrival Rate ($a_j^{(d)} \in \mathbb{R}$):} The aggregate flow of UAVs with destination $d$ arriving at vertiport $j$ from all incoming links, defined by ~\cite{weiss2021scheduling} as:
    \begin{equation}
    a_j^{(d)}(t) := \sum_{i \neq j} \mu_{ij} f_{ij}^{(d)}(t)
    \label{eq: Arrival Rate Definition}
    \end{equation}
    \item \textbf{Service Start Rate ($r_j(t) \in \mathbb{R} $):} The rate at which UAVs with destination $d$ enter service at vertiport $j$, limited by the available service capacity $\kappa_j$~\cite{weiss2021scheduling}. \\
    
    The service start rate depends on whether a queue is present. Specifically:
    \begin{equation}
    r_j(t) =
    \begin{cases}
    \mu_s \kappa_j, & \text{if } q_j(t) > 0, \\[6pt]
    \min \{ a_j(t), \mu_s \kappa_j \}, & \text{if } q_j(t) = 0.
    \end{cases}
    \end{equation}
    
    This can be written more compactly as:
    \begin{equation}
    r_j(t) = \min \{ a_j(t), \mu_s \kappa_j \}.
    \label{eq: Service Start Rate}
    \end{equation}
    
    We define $\{ r_j^{(d)}(t) \}_{d=1}^V$ as the class-specific service start rates satisfying
    \begin{equation}
    \sum_d r_j^{(d)}(t) = r_j(t).
    \end{equation}
    
    Due to the nonlinearity of the minimum operator, $r_j^{(d)}(t)$ does not admit a simple closed-form expression in terms of $a_j^{(d)}(t)$.
    
    \item \textbf{Routing ($p_{ij}^{(d)}$):} The routing matrix entry $p_{ij}^{(d)}$ represents the probability that a UAV heading to destination $d$ will head to vertiport $j$ for their next hop, given they are at vertiport $i$~\cite{weiss2021scheduling}. The routing matrix satisfies the following conditions:
    \begin{equation}
    p_{ij}^{(d)} \geq 0, \quad
    \sum_j p_{ij}^{(d)} = 1,\, \forall i, \quad
    p_{ii}^{(d)} = 0,\, \forall i
    \label{eq:Routing Matrix Definition}
    \end{equation}
    
\end{itemize}
This fluid regime allows for the analysis of steady-state flow equilibria and network stability without the dimensionality constraints of the underlying CTMC~\cite{weiss2021scheduling}.

While the fluid approximation is formally justified in the $N \to \infty$ limit~\cite{kurtz1970solutions}, any real UAM deployment operates with a finite fleet size $N$; near-term deployments are expected to involve fleet sizes on the order of hundreds of UAVs within a single metropolitan area~\cite{nasa}. For finite $N$, the discrepancy between the stochastic trajectory $X(t)/N$ and its fluid limit $x(t)$ is known to concentrate around zero with fluctuations of order $O(\frac{1}{\sqrt{N}})$~\cite{kurtz1970solutions}. The fluid trajectory $x(t)$ therefore provides a high-probability, uniformly accurate surrogate for $X(t)/N$, with the approximation error shrinking as the fleet scales, which justifies treating the ODE system~\ref{eq: Dynamical System} as the operative model for routing design.

\subsection{Optimal Routing Formulation}
This section defines the continuous-time average-cost optimal control problem used to determine the optimal routing policy.

\subsubsection{Decision Variables}
The primary goal is to find the optimal destination-conditioned routing policy set $\mathcal{P}^* = \{P^{(1)}, P^{(2)}, \dots, P^{(d)}\}$, where $P^{(1)}, P^{(2)},...,P^{(d)} \in \mathbb{R}^{N \times N}$ for each destination class $d \in \mathcal{V}$. The individual decision variables are the transition probabilities $p_{ij}^{(d)}$, which govern the movement of UAVs from vertiport $i$ to $j$ for a specific destination $d$.

\subsubsection{Objective Function}
The optimization seeks to minimize the long-run average cost per each hop $i\to j$. In particular, let:
\begin{equation}
\beta_{ij}^{(d)} := \tau_{ij} + h \cdot \mathbf{1}_{[j \neq d]} \label{eqref: Cost function}
\end{equation}
be the cost incurred when making hop $ i \to j $. Here, $\tau_{ij} = \frac{\bar{v}}{l_{ij}}$ is the average time to go from $i \to j$ (defined as the average velocity $\bar{v}$ divided by the route length $l_{ij}$ between $i$ and $j$) and $h \cdot \mathbf{1}[j \neq d]$ is the cost of visiting hops $j$ that are not equal to the final destination $d$.

For a destination class \(d\), let \(Y_n^{(d)}\) denote the chain of hops between the vertiports induced by the routing matrix $P^{(d)}$. Assume the chain is ergodic 
then, over the first \(N\) hops between vertiports, the total accumulated cost is:
\[
C_N^{(d)} = \sum_{n=0}^{N-1} \beta_{Y_n^{(d)},Y_{n+1}^{(d)}}.
\]
Hence, we define the objective function as the the long-run average cost per jump:
\begin{align}
J = \sum_d m_d J^{(d)}, \quad
   J^{(d)} = \lim_{N\to\infty}
\frac{1}{N}\,\mathbb{E}\! \left[C_N^{(d)}\right], 
\label{eq:obj_function}
\end{align}
where $m_d$ denotes mass of the UAV flow with a destination class \(d\).


\subsubsection{System Constraints}
For the routing policy to be valid, the following constraints must be satisfied for all $i, d \in \mathcal{V}$ and all $t \ge 0$:
\begin{itemize}
    \item \textbf{Network Dynamics:} The system must adhere to the destination-conditioned queue, service, and flight dynamics:
    \begin{subequations}
    \begin{align}
        \dot{q}_{i}^{(d)}(t) &= \sum_{k \neq i} \mu_{ki} f_{ki}^{(d)}(t) - r_{i}^{(d)}(t) \\
        \dot{s}_i^{(d)}(t) &= r_i^{(d)}(t) - \mu_s s_i^{(d)}(t) \\
        \dot{f}_{ij}^{(d)}(t) &= \mu_s p_{ij}^{(d)}(t) s_i^{(d)}(t) - \mu_{ij} f_{ij}^{(d)}(t)
    \end{align}
    \end{subequations}
    
    \item \textbf{Routing Integrity:} The decision variables must satisfy $\sum_{j \in \mathcal{V}} p_{ij}^{(d)} = 1$, $p_{ij}^{(d)} \ge 0$, and $p_{ii}^{(d)} = 0$.
    
    \item \textbf{Mass Conservation:} The total mass for each destination class $m_d$ is fixed:
    \begin{equation}
        \sum_{i \in \mathcal{V}} \left(q_i^{(d)}(t) + s_i^{(d)}(t)\right) + \sum_{i,j \in \mathcal{V}} f_{ij}^{(d)}(t) = m_d. 
        \label{eqref: Mass conservation w/ Q}
    \end{equation}

    \item \textbf{Safety Constraint:} To ensure safe operation and prevent congestion, we enforce $q_i^{(d)}(t) = 0$ for all $t \ge 0$.
\end{itemize}

Solving this optimization problem is computationally intractable due to the infinite-dimensional constraints imposed by the continuous-time dynamics and the non-linear coupling between the routing variables and fleet distribution. The requirement for a zero-queue state at all times further restricts the feasible set, necessitating a rigorous analysis of transient behavior of the network to derive tractable control laws for large-scale UAM systems.


\section{Queue-Free Analysis and Stability}
\label{sec:Queue-Free Analysis and Stability}
This section analyzes the conditions under which the UAM network operates in an underloaded regime, characterized by the elimination of aerial queues. Our analysis follows a two-step approach. First, we establish the sufficient conditions for queue-free operation. Next, we show that the system dynamics are positively invariant with respect to the queue-free state; therefore, if the network is initialized without queues and satisfies these queue-free conditions, it remains queue-free indefinitely.

\subsection{Sufficient condition for Queue-free States}

To ensure safe operation, we need to maintain queue-free states for our system. This subsection establishes the sufficient conditions for queue-free conditions. 

\begin{proposition}
\label{prop:queue_free}
Suppose that:
\begin{equation}
a_i(t) < \mu_s \kappa_i, \qquad \forall i \in \mathcal V,\ \forall t \ge 0,
\label{eq:Q-free-suff-condition}
\end{equation}
where:
\[
a_i(t)=\sum_d a_i^{(d)}(t).
\]
Then:

\begin{enumerate}
\item if $q_i(0)=0$, then $q_i(t)=0$ for all $t\ge 0$;
\item if $q_i(0)>0$, then $q_i(t)$ decreases whenever $q_i(t)>0$, and hence the queue is eventually drained to zero, after which it remains zero.
\end{enumerate}
\end{proposition}

The proof is given in Appendix~\ref{app:queue_free}. This condition ensures that limiting our arrival rate is sufficient to process all incoming flow, preventing queue formation.

\subsection{Positive Invariance of Queue-Free Condition}

The goal of this subsection is to prove the positive invariance our of queue-free state space. We prove that if $ a_i (0) < \mu_s \kappa_i \Rightarrow a_i(t) < \mu_s \kappa_i,\, \forall t \geq 0 $.
%
%
%
To that end, we define the safe set as $\mathcal{X} = \{x \ge 0 : \sum x = 1, s_i \le \kappa_i, a_i < \mu_s \kappa_i\}$.  Next, we show that this set is positively invariant under the network dynamics as captured by the following result whose proof is given in Appendix~\ref{app:pos_inv}.

\begin{proposition}
\label{prop:pos_inv}
Consider the set:
\begin{equation}
\mathcal X
=
\left\{
x \ge 0 :
\mathbf 1^\top x = 1,\;
s_i \le \kappa_i,\;
a_i < \mu_s \kappa_i,\;
\forall i\in\mathcal V
\right\}
\label{eq: Invariant Set}
\end{equation}
where:
\[
s_i(t) := \sum_d s_i^{(d)}(t),
\qquad
a_i(t) := \sum_d \sum_{k\neq i} \mu_{ki} f_{ki}^{(d)}(t).
\]
Then $\mathcal X$ is positively invariant under the destination-conditioned dynamics:
\begin{align}
\dot{s}_i^{(d)}(t)
&=
\sum_{k\neq i}\mu_{ki} f_{ki}^{(d)}(t)-\mu_s s_i^{(d)}(t), \label{eq:pos_dyn_1}\\ 
\dot{f}_{ij}^{(d)}(t)
&=
\mu_s p_{ij}^{(d)} s_i^{(d)}(t)-\mu_{ij} f_{ij}^{(d)}(t),
\qquad i\neq j. \label{eq:pos_dyn_2}
\end{align}
\end{proposition}

\section{Refined Optimization Problem}
\label{sec:Refined Optimization Problem}

This section shows several simplifications of the original optimization problem. First, we show how to simplify the objective function in~\eqref{eq:obj_function} to a simpler form. Next, we will exploit the positive invariance of the queue-free state space proof from Proposition 2 to remove the infinite-dimensional differential equations and simplify it to finite-dimensional algebraic constraints.

\subsection{Objective Function Derivation}

This subsection derives our optimization problem that determines an optimal routing policy $P = \{P^{(1)}, P^{(2)}, \dots, P^{(d)}\}$ that jointly optimizes for travel time,
multi-hop efficiency, unmet demand, and vertiport congestion. First, we derive the long-run average cost per jump for an ergodic chain with destination-conditioned transition matrix $P^{(d)}$. Then, we define a cost function that penalizes travel time and unnecessary hops. We substitute this cost function into our expected one-jump cost to derive our objective function. Finally, we derive our constraints that ensure routing integrity, steady-state flow, mass conservation, and safe capacity.  

\begin{proposition}
    \label{prop:obj_func_derivation}
    Minimizing the objective function $J$ in~\eqref{eq:obj_function} is equivalent to minimizing:
    \begin{align}
        J =\sum_{d,i,j} m_d\,\pi_i^{(d)} p_{ij}^{(d)} \tau_{ij}-h\sum_d m_d \pi_d^{(d)}.
        \label{eq:obj_function2}
    \end{align}
\end{proposition}

This proof is given in Appendix~\ref{app:obj_func_derivation}.

\subsection{Algebraic Constraints Derivation}
The utility of the positive invariance result (Proposition 2) lies in its ability to decouple the long-run optimal routing problem from the complexities of transient state trajectories. we guarantee that any system initialized within this region will never exit it under the prescribed dynamics. This mathematical assurance allows us to replace the infinite-dimensional, time-varying differential constraints of the original optimal control problem with a set of static, algebraic steady-state equations. Consequently, we can ensure network-wide safety and stability across the entire time horizon solely by satisfying underloaded capacity conditions at equilibrium, effectively eliminating the need for computationally intensive numerical integration or explicit monitoring of transient behavior.



\paragraph{Steady-State Analysis}

We analyze the network at its equilibrium point $x^* = (q^*, s^*, f^*)$, defined by $\dot{q}=\dot{s}=\dot{f}=0$. Doing so introduces three important cases: Under loaded (Stable) case, Critical case, and the Overloaded case. We are interested in maintaining the under loaded case.

\textbf{Underloaded Case}
Setting $\dot{q}_j^{(d)} = 0$ implies $a_j^{(d)*} = r_j^{(d)*}$, and $\dot{s}_j^{(d)} = 0$ yields $r_j^{(d)*} = \mu_s s_j^{(d)*}$. Hence, at steady state:
\begin{align}
s_j^{(d)*} = \frac{a_j^{(d)*}}{\mu_s}. \label{eq: s_equil}
\end{align}

If we impose $ a_j^{(d)*} < \mu_s c $, then by \eqref{eq: s_equil} we have $s_j^{(d)*} < c$. This condition, in turn ensures that the pads are under capacity $c$ and hence it naturally follows that all queues will be empty at steady state, i.e., $s_j^{(d)*} < c \Rightarrow q_j^{(d)*} = 0 $. 
Similarly, from $\dot{f}_{ij}^{(d)} = 0$, we conclude that:
\begin{equation}
    f_{ij}^{(d)*} = \frac{\mu_s p_{ij}^{(d)} s_i^{(d)*}}{\mu_{ij}}. \label{eq: f_equil}
\end{equation}
Substituting $s_j^{(d)*}$ from \eqref{eq: s_equil} yields: 
\begin{align*}
f_{ij}^{(d)*} = \frac{p_{ij}a_i^{(d)*}}{\mu_{ij}}.
\end{align*}
Next, we substitute the above equation for $f_{ij}^{(d)*}$ into $ a_j^{(d)}(t) = \sum_i \sum_{j \neq i} \mu_{ij} f_{ij}^{(d)}(t) $ to get:
\begin{align}
a_j^{(d)*} &= \sum_{i \neq j} \mu_{ij} f_{ij}^{(d)*} = \sum_{i \neq j} \mu_{ij} \left( \frac{p_{ij}a_i^{(d)*}}{\mu_{ij}} \right) \\
&= \sum_{i\neq j} p_{ij} a_i^{(d)*} \label{eq: a_j = p_ij a_i}.
\end{align}
Finally, equation \eqref{eq: a_j = p_ij a_i} simplifies to: 
\begin{equation}
a^{(d)*} = a^{(d)*} P^{(d)}. \label{eqref: a=aP}
\end{equation}
In other words, at steady state, the arrival rate $a^{(d)*}$ is equal to the left eigenvector of the routing matrix $P^{(d)}$. We are now ready to derive our constraints

\paragraph{Eliminating flight variables}

We derive steady-state constraints that eliminate the flight variables \(f_{ij}^{(d)}\) in favor of the service variables \(s_i^{(d)}\). 
Note for readability, we use $ f_{ij}^{(d)} = f_{ij}^{(d)*} $ and $ s_{ij}^{(d)} = s_{ij}^{(d)*} $.

At equilibrium, the flight dynamics satisfy Equation \eqref{eq: f_equil}:
which implies:
\[
f_{ij}^{(d)}
=
\frac{\mu_s}{\mu_{ij}} p_{ij}^{(d)} s_i^{(d)}.
\]
Thus, all flight flows can be expressed as linear functions of \(s\) and \(P^{(d)}\).

\paragraph{Steady-state flow balance}
The arrival rate into node \(i\) for class \(d\) is
\[
a_i^{(d)}
=
\sum_{k\neq i} \mu_{ki} f_{ki}^{(d)}.
\]
Substituting the steady-state expression for \(f_{ki}^{(d)}\):
\[
a_i^{(d)}
=
\sum_{k\neq i} \mu_{ki}
\left(
\frac{\mu_s}{\mu_{ki}} p_{ki}^{(d)} s_k^{(d)}
\right)
=
\mu_s \sum_{k\neq i} p_{ki}^{(d)} s_k^{(d)}.
\]
At steady state, \(a_i^{(d)} = \mu_s s_i^{(d)}\), hence:
\[
\sum_{k} p_{ki}^{(d)} s_k^{(d)} = s_i^{(d)}.
\]
This yields the flow conservation constraint:
\begin{equation}
\sum_{j} p_{ji}^{(d)} s_j^{(d)} - s_i^{(d)} = 0.
\label{eq: Flow Conservation Constraint}
\end{equation}

In vector form:
\begin{equation}
s^{(d)} = s^{(d)} P^{(d)}.
\end{equation}

Since \(P^{(d)}\) is irreducible, there exists a unique stationary distribution:
\[
\pi^{(d)} = \pi^{(d)} P^{(d)},
\qquad
\pi^{(d)} \ge 0,
\qquad
\sum_i \pi_i^{(d)} = 1.
\]
Thus, any steady-state solution must take the form:
\begin{equation}
s_i^{(d)} = \alpha_d \, \pi_i^{(d)},
\label{eq: s_pi_param}
\end{equation}
where \(\alpha_d = \sum_i s_i^{(d)} > 0\) is a scalar determined by mass conservation.

\paragraph{Mass conservation}
Using the steady-state expression for \(f_{ij}^{(d)}\), the total class mass is:
\[
m_d
=
\sum_i s_i^{(d)}
+
\sum_{i\neq j} f_{ij}^{(d)}
=
\sum_i s_i^{(d)}
+
\sum_i \sum_{j\neq i}
\frac{\mu_s}{\mu_{ij}} p_{ij}^{(d)} s_i^{(d)}.
\]
Substituting \eqref{eq: s_pi_param}, we obtain:
\[
m_d
=
\alpha_d
\left[
\sum_i \pi_i^{(d)}
+
\sum_i \pi_i^{(d)} \sum_{j\neq i}
\frac{\mu_s}{\mu_{ij}} p_{ij}^{(d)}
\right].
\]
Using \(\sum_i \pi_i^{(d)} = 1\), this simplifies to:
\begin{equation}
m_d
=
\alpha_d
\left(
1 +
\sum_i \pi_i^{(d)} \sum_{j\neq i}
\frac{\mu_s}{\mu_{ij}} p_{ij}^{(d)}
\right).
\end{equation}
Solving for \(\alpha_d\), we obtain:
\begin{equation}
\alpha_d
=
\frac{m_d}{
1 +
\sum_i \pi_i^{(d)} \sum_{j\neq i}
\frac{\mu_s}{\mu_{ij}} p_{ij}^{(d)}
}.
\label{eq: alpha_d Definition}
\end{equation}

\paragraph{Combined characterization}
Combining flow conservation and mass conservation, the steady-state solution is fully characterized by:
\begin{equation}
s_i^{(d)}
=
\alpha_d \, \pi_i^{(d)},
\qquad
\pi^{(d)} = \pi^{(d)} P^{(d)},
\end{equation}
with \(\alpha_d\) given by \eqref{eq: alpha_d Definition}.
In this combined characterization, the steady-state constraint contributes the shape $ \pi^{(d)} $ and mass conservation contributes to the scale $ \alpha_d $.

\paragraph{Capacity constraints}
The aggregate service load at vertiport \(i\) is:
\[
s_i = \sum_d s_i^{(d)}.
\]
To ensure queue-free operation, we require:
\begin{equation}
\label{eq: service_constraint_1}
\sum_d s_i^{(d)} < \kappa_i.
\end{equation}
Similarly, the total incoming service-induced flow satisfies:
\begin{equation}
\label{eq: service_constraint_2}
\sum_{d,k} p_{ki}^{(d)} s_k^{(d)} < \kappa_i,
\end{equation}
which ensures that the induced arrival rate does not exceed service capacity.

\paragraph{Routing constraints}

Recall the definition of routing matrix $P^{(d)} = [p_{ij}^{(d)}]$ from definition \eqref{eq:Routing Matrix Definition}:

\begin{equation}
p_{ij}^{(d)} \geq 0, \quad
\sum_j p_{ij}^{(d)} = 1,\, \forall i, \quad
p_{ii}^{(d)} = 0,\, \forall i
\end{equation}

\begin{equation}
p_{ij}^{(d)} \geq \epsilon \cdot A_{ij}, \quad \forall i, j \in \mathcal{V}\label{eq:irreducibility_constraint}
\end{equation}
where $A_{ij} \in \{0, 1\}$ is the $(i,j)$-th entry of the adjacency matrix $A$ corresponding to the physical network graph $\mathcal{G}$. Specifically, $A_{ij} = 1$ if a valid flight link exists from vertiport $i$ to $j$, and $A_{ij} = 0$ otherwise. We assume the underlying physical network topology is strongly connected, meaning there exists a valid directed flight path between any pair of nodes in $\mathcal{V}$. Note that $A$ is independent of the destination class $d$, reflecting the shared physical infrastructure.

The parameter $\epsilon > 0$ is a strictly positive, arbitrarily small tolerance. By lower-bounding the transition probabilities along all valid physical edges, constraint \eqref{eq:irreducibility_constraint} guarantees that the destination-conditioned routing matrix $P^{(d)}$ characterizes an irreducible Markov chain. This irreducibility is critical, as it ensures the existence of a unique stationary distribution for the network flows.

\subsection{Final Optimization Problem}
At steady state, all constraints can be expressed solely in terms of \(s_i^{(d)}\) and \(p_{ij}^{(d)}\). The flight variables \(f_{ij}^{(d)}\) are completely eliminated via the equilibrium relation, reducing the problem to the following finite-dimensional algebraic system.


\begin{align*}
\min_{P^{(d)}, s} \quad & \sum_{d,i,j} m_d \pi_i^{(d)} p_{ij}^{(d)} \tau_{ij} - h \sum_{d} m_d \pi_d^{(d)}, \
\text{s.t.} \\
\quad & \sum_{j} p_{ij}^{(d)} = 1, \quad p_{ij}^{(d)} \ge 0, \quad p_{ii}^{(d)} = 0, \\ &P_{ij}^{(d)} \geq \epsilon \cdot A_{ij},  \quad \forall i, j \in V  \quad (\text{Routing Integrity})  \\
& s^{(d)} = \alpha_d \pi^{(d)}, \quad \\ &\quad \; \qquad (\text{Steady-State Flow and Mass Conservation})  \\
& \sum_d s_i^{(d)} < \kappa_i, \quad \sum_{d,k} p_{ki}^{(d)} s_k^{(d)} < \kappa_i \quad (\text{Safe Capacity}) 
\end{align*}

This static formulation eliminates the need for time-dependent integration, ensuring global safety through steady-state algebraic constraints.

Substituting $\tau_s := 1/\mu_s$ and $\tau_{ij} := 1/\mu_{ij}$, this expression simplifies to
\begin{equation}
\alpha_d = m_d \, \frac{\tau_s}{\tau_s + \mathbb{E}[\tau_{ij}]},
\label{eq:alpha_d Simplified}
\end{equation}
where $\mathbb{E}[\tau_{ij}] := \sum_{i,j} \pi_i^{(d)} p_{ij}^{(d)} \tau_{ij}$ is the mean one-hop flight time under the stationary routing behavior of class $d$, consistent with the one-jump expectation used in Proposition~\ref{prop:obj_func_derivation}. Intuitively, $\alpha_d$ is the fraction of the total mean cycle time $\tau_s + \mathbb{E}[\tau_{ij}]$ spent in service, scaled by the total class mass $m_d$: when UAVs spend proportionally more time in flight relative to service ($\mathbb{E}[\tau_{ij}] \gg \tau_s$), $\alpha_d$ shrinks, reflecting that a smaller share of the fleet is occupied at service pads at any given time, which loosens the service-capacity constraints~\eqref{eq: service_constraint_1} and~\eqref{eq: service_constraint_2} and generates larger capacity $\kappa$.

\section{Numerical Examples}

We present five numerical experiments validating the theoretical properties 
of the routing optimization. First, we examine the effect of the hop penalty 
$h$ on both a symmetric network (uniform edge rates) and a heterogeneous 
network (varied $\mu_{ij}$), demonstrating that $h$ shifts the objective 
value without altering the optimal policy in the symmetric case, while 
inducing a meaningful Pareto trade-off between travel efficiency and 
destination reachability in the heterogeneous case. Next, we sweep pad 
capacity $\kappa_i$ on a $V = 3$ vertiport heterogeneous network, identifying 
a hard feasibility floor below which no valid routing policy exists and a 
saturation threshold $\kappa^*$ above which the capacity constraint becomes 
non-binding. Additionally, we observe the effect of the parameters $\mu_s$ and $\mu_{ij}$ on the capacity $\kappa$. Finally, we compare our method with two baseline methods and compare the costs and minimum capacity $\kappa_{\min}$.

\subsection{Effect of hop penalty $h$}

We examine how the hop penalty $h$ shapes the optimal routing policy across
symmetric and heterogeneous networks, demonstrating that $h$ controls the
trade-off between minimizing travel time and maximizing destination
reachability.

\paragraph{Symmetric network ($\mu_{ij} = 1$ for all $i \neq j$).}
As established analytically, the weighted travel time term
$\sum_{d,i,j} m_d \pi_i^{(d)} P_{ij}^{(d)} \tau_{ij}$ is identically equal
to 1 for all $h$, since uniform edge rates imply $\tau_{ij} = 1$ for all
$i \neq j$ and the weighted average collapses to a constant regardless of
the routing policy.
Consequently, the total objective decreases linearly with $h$ at rate
$-\sum_d m_d \pi_d^{(d)} \approx -1/2$, and the reachability term grows
linearly at the same rate as $\pi_d^{(d)}$ saturates at its structural
upper bound of $1/2$ immediately at $h = 0^+$.
The optimal routing policy is thus determined entirely by the reachability
term and does not change with $h$ --- only the objective value shifts,
confirming that $h$ has no effect on the optimal $P^{(d)}$ once
saturation is achieved.

\begin{figure}[htbp]
    \centering
    \includegraphics[width=\columnwidth]{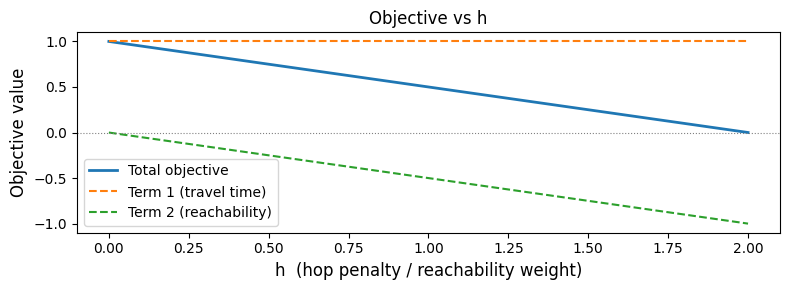}
    \caption{Objective function vs. h [Symmetric network]}
    \label{fig:symmetric_h}
\end{figure}

\paragraph{Heterogeneous network ($\mu_{ij} \in [0.7,\, 1.25]$).}
In contrast to the symmetric case, the travel time term is no longer
constant: it begins at approximately $0.91$ at $h = 0$ and rises to
$\approx 1.0$ as $h$ increases, reflecting a trade-off between
routing efficiency and destination reachability.
At $h = 0$, the optimizer freely exploits the fastest edges without
regard for reaching destinations, achieving a lower travel time cost;
as $h$ grows, the routing policy is pulled toward destination-directed
paths which are not always the shortest, causing the travel time term
to increase.
The reachability term again grows approximately linearly once
$\pi_d^{(d)}$ saturates at $1/2$, but the saturation occurs more
gradually than in the symmetric case, which is visible as a slight inflection
in the travel time term near $h \approx 1$. This confirms that edge
heterogeneity induces a meaningful Pareto frontier between travel
efficiency and destination reachability that is absent in the symmetric
network.

\begin{figure}[htbp]
    \centering
    \includegraphics[width=\columnwidth]{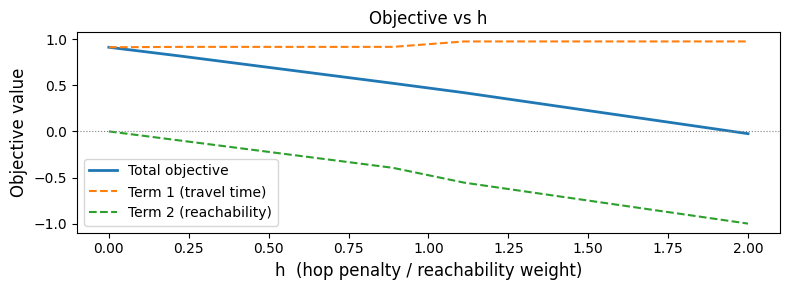}
    \caption{Objective function vs. h [Heterogeneous network]}
    \label{fig:unsymmetric_h}
\end{figure}


\subsection{Effect of Capacity $\kappa$ on Optimal Routing}

\label{subsec: Num Results B}

We investigate how pad capacity $\kappa_i$ constrains the feasible routing
space, identifying a hard feasibility floor below which no valid policy
exists and a saturation threshold $\kappa^*$ beyond which further capacity
increases yield no improvement in the optimal objective.

\begin{figure}[htbp]
    \centering
    \includegraphics[width=\columnwidth]{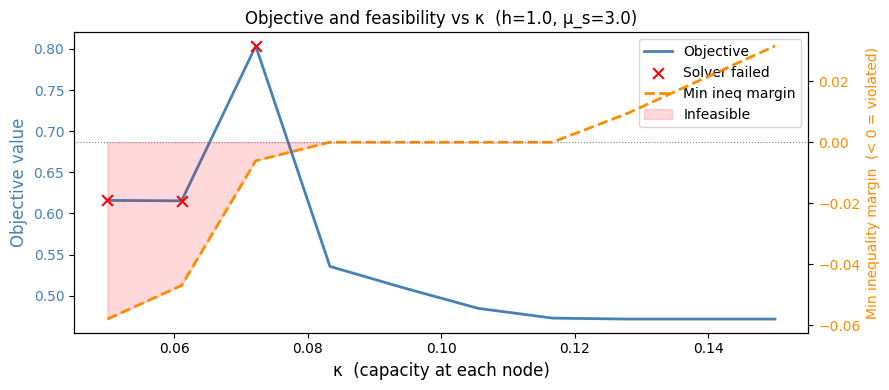}
    \caption{Objective value and Feasibility vs. Kappa values}
    \label{fig:kappa_sweep}
\end{figure}

We consider a $V = 3$ vertiport network with heterogeneous flight rates
$\mu_{ij} \in [0.7, 1.25]$, uniform service rate $\mu_s = 3.0$, hop penalty $h = 1.0$,
and uniform class masses $m_d = 1/3$.
Solving the routing optimization across $\kappa \in [0.05, 0.15]$
reveals a threshold effect: below a hard feasibility floor
$\kappa_{\min} \approx 0.04$, no routing policy satisfies the capacity
constraint and the solver fails; in the transition region
$\kappa \in [0.08, 0.12]$, the binding constraint forces suboptimal
routing and the objective degrades sharply from $0.536$ to $0.472$;
above the saturation threshold $\kappa^* \approx 0.12$, the constraint
becomes non-binding and the objective plateaus at $0.472$ regardless of
further capacity increases.
This identifies $\kappa^*$ as the minimum pad capacity required to
achieve unconstrained-optimal routing, providing a principled
infrastructure design target for the network.



\subsection{Effect of $\mu_s$ and $\mu_{ij}$ on $\kappa$}


Equation~\eqref{eq:alpha_d Simplified} demonstrates how the parameters $\mu_s$ and $\mu_{ij}$ have an effect on the safe capacity margins. To corroborate this, we used a similar setup to Section~\ref{subsec: Num Results B}, except we varied $\mu_s$ to change the value of $\frac{\mu_s}{\mu_{ij}}$, and we observed its effect on $\kappa_{\min}$. When $\frac{\mu_s}{\mu_{ij}}$ values were $2.0$, $3.0$, and $4.0$, we observed $\kappa_{\min}$ to be $0.106$, $0.083$, and $0.061$ respectively.

From the results, it is clear that as we increase $\frac{\mu_s}{\mu_{ij}}$ ratio, the $\kappa_{\min}$ values decrease. This is because when service rate is high relative to flight rates, then a smaller share of the fleet is occupying vertiports. Thus, we can have a larger capacity margin.

\subsection{Comparison with baseline methods}

We compare our optimization technique to two baseline methods. First, we compare it to a routing policy in which each vertiport distributes outgoing flow uniformly accross all valid links: $p_{ij}^{(d)} = 1/|\mathcal{N}(i)|$ for all $j \in \mathcal{N}(i)$ denotes the set of neighbors reachable from $i$ under the network topology described by $A$.

Second, we consider a greedy routing policy in which each vertiport routes all outgoing flow to the single neighbor minimizing travel time toward the destination class: $p_{ij}^{(d)} = 1$ for $j = \arg \min_{k \in \mathcal{N}(i)} \tau_{ik}$ and $p_{ij}^{(d)} = 0$ otherwise. 

We consider a $V = 3$ vertiport network with uniform flight rates $\mu_{ij} = 1.0$, uniform service rate $\mu_s = 2.0$, hop penalty $h = 1.0$, and uniform class masses $m_d = 1/3$. Solving the optimized routing compared to the two baseline methods highlights the effectiveness of the method: From Table~\ref{table:routing_comparison}, the optimal routing strategy achieves the minimum cost with minimum capacity required to be feasible.

\begin{table}[!t]
\renewcommand{\arraystretch}{1.3}
\caption{Comparison of routing strategies and their cost and minimum capacity}
\label{table:routing_comparison}
\centering
\begin{tabular}{c||c|c|c}
\hline
 & \bfseries Shortest Path & \bfseries Equal Split & \bfseries Ours\\
\hline\hline
$\text{Objective value}~J$ & 0.5607 & 0.6225 & \textbf{0.4477}\\
\hline
$\text{Minimum capacity}~\kappa_{\min}$ & 0.111 & 0.075 & \textbf{0.075}\\
\hline
\end{tabular}
\end{table}

\section{Conclusion}
This paper introduced a tractable, destination-conditioned fluid framework for the optimal routing of UAVs in high-density UAM networks. By employing a mean-field approximation of the underlying CTMC, we derived a dynamical system capable of modeling complex vertiport and flight interactions. The central contribution of this work is the formal proof of the positive invariance of the queue-free state space. This result provides a rigorous justification for simplifying the infinite-dimensional optimal control problem into a tractable, finite-dimensional algebraic optimization. 

The resulting framework ensures network stability and safety-critical operation (i.e., zero queues) across the entire time horizon solely by satisfying underloaded capacity conditions at steady state. Our analysis demonstrates that the refined optimization problem effectively balances travel efficiency with destination-oriented flow. Future research will investigate the robustness of these fluid-based routing policies under stochastic demand fluctuations and integrate these high-level routing strategies with low-level conflict resolution and avoidance maneuvers in decentralized UAM environments.

\bibliographystyle{IEEEtran}
\bibliography{references}

\begin{thebibliography}{10}
\providecommand{\url}[1]{#1}
\csname url@samestyle\endcsname
\providecommand{\newblock}{\relax}
\providecommand{\bibinfo}[2]{#2}
\providecommand{\BIBentrySTDinterwordspacing}{\spaceskip=0pt\relax}
\providecommand{\BIBentryALTinterwordstretchfactor}{4}
\providecommand{\BIBentryALTinterwordspacing}{\spaceskip=\fontdimen2\font plus
\BIBentryALTinterwordstretchfactor\fontdimen3\font minus \fontdimen4\font\relax}
\providecommand{\BIBforeignlanguage}[2]{{%
\expandafter\ifx\csname l@#1\endcsname\relax
\typeout{** WARNING: IEEEtran.bst: No hyphenation pattern has been}%
\typeout{** loaded for the language `#1'. Using the pattern for}%
\typeout{** the default language instead.}%
\else
\language=\csname l@#1\endcsname
\fi
#2}}
\providecommand{\BIBdecl}{\relax}
\BIBdecl

\bibitem{bharadwaj2021decentralized}
S.~Bharadwaj, S.~Carr, N.~Neogi, and U.~Topcu, ``Decentralized control synthesis for air traffic management in urban air mobility,'' \emph{IEEE Transactions on Control of Network Systems}, vol.~8, no.~2, pp. 598--608, 2021.

\bibitem{thibbotuwawa2020unmanned}
A.~Thibbotuwawa, G.~Bocewicz, P.~Nielsen, and Z.~Banaszak, ``Unmanned aerial vehicle routing problems: A literature review,'' \emph{Applied sciences}, vol.~10, no.~13, p. 4504, 2020.

\bibitem{zhang2022method}
H.~Zhang, Y.~Fei, J.~Li, B.~Li, and H.~Liu, ``Method of vertiport capacity assessment based on queuing theory of unmanned aerial vehicles,'' \emph{Sustainability}, vol.~15, no.~1, p. 709, 2022.

\bibitem{han2026rolling}
H.~Han and B.~D. Song, ``Rolling horizon optimization of urban air mobility (uam) service with shared riding, vertiport-airspace capacity, and recharge: a mathematical model and efficient heuristic,'' \emph{Transportation Research Part E: Logistics and Transportation Review}, vol. 206, p. 104548, 2026.

\bibitem{zheng2021optimal}
W.~Zheng, P.~Thangeda, Y.~Savas, and M.~Ornik, ``Optimal routing in stochastic networks with reliability guarantees,'' in \emph{2021 IEEE International Intelligent Transportation Systems Conference (ITSC)}.\hskip 1em plus 0.5em minus 0.4em\relax IEEE, 2021, pp. 3521--3526.

\bibitem{vascik2019aiaa}
P.~D. Vascik and R.~J. Hansman, ``Development of vertiport capacity envelopes and analysis of their sensitivity to topological and operational factors,'' in \emph{AIAA Scitech 2019 Forum}, 2019, p. 0526.

\bibitem{ross2019probability}
S.~M. Ross, \emph{Introduction to Probability Models}, 12th~ed.\hskip 1em plus 0.5em minus 0.4em\relax Academic Press, 2019.

\bibitem{kolmogoroff1931ueber}
A.~Kolmogoroff and A.~Kolmogorov, ``Ueber die analytischen methoden in der wahrscheinlichkeitstheorie,'' \emph{Math. Ann.}, vol. 104, pp. 415--458, 1931.

\bibitem{ethier2009markov}
S.~N. Ethier and T.~G. Kurtz, \emph{Markov processes: characterization and convergence}.\hskip 1em plus 0.5em minus 0.4em\relax John Wiley \& Sons, 2009.

\bibitem{harchol2013performance}
M.~Harchol-Balter, \emph{Performance modeling and design of computer systems: queueing theory in action}.\hskip 1em plus 0.5em minus 0.4em\relax Cambridge University Press, 2013.

\bibitem{kurtz1970solutions}
T.~G. Kurtz, ``Solutions of ordinary differential equations as limits of pure jump markov processes,'' \emph{Journal of applied Probability}, vol.~7, no.~1, pp. 49--58, 1970.

\bibitem{weiss2021scheduling}
G.~Weiss, \emph{Scheduling and Control of Queueing Networks}.\hskip 1em plus 0.5em minus 0.4em\relax Cambridge University Press, 2021, vol.~14.

\bibitem{nasa}
\BIBentryALTinterwordspacing
M.~D. Patterson, D.~R. Isaacson, N.~L. Mendonca, N.~A. Neogi, K.~H. Goodrich, M.~Metcalfe, W.~Bastedo, C.~Metts, B.~P. Hill, D.~DeCarme, C.~Griffin, and S.~Wiggins, ``An initial concept for intermediate-state, passenger-carrying urban air mobility operations,'' in \emph{AIAA SciTech 2021 Forum}, 2021, p. 1626. [Online]. Available: \url{https://nasa.gov}
\BIBentrySTDinterwordspacing

\end{thebibliography}

\appendices
\section{Proof of Proposition~\ref{prop:queue_free}}
\label{app:queue_free}
\begin{proof}
Recall from \eqref{eq: Service Start Rate} that the service start rate is:
\[
r_i(t)=
\begin{cases}
\mu_s \kappa_i, & q_i(t)>0,\\[4pt]
\min\{a_i(t),\mu_s\kappa_i\}, & q_i(t)=0.
\end{cases}
\]
The queue dynamics from \eqref{eq: Dynamical System} are:
\[
\dot q_i(t)=a_i(t)-r_i(t).
\]
If $q_i(t)>0$, then:
\[
\dot q_i(t)=a_i(t)-\mu_s\kappa_i<0,
\]
by \eqref{eq:Q-free-suff-condition}. Hence any positive queue decreases.
If $q_i(t)=0$, then since $a_i(t)<\mu_s\kappa_i$,
\[
r_i(t)=\min\{a_i(t),\mu_s\kappa_i\}=a_i(t),
\]
so:
\[
\dot q_i(t)=a_i(t)-r_i(t)=0.
\]
Thus, once the queue reaches zero, it remains zero.
Therefore, if $q_i(0)=0$, then $q_i(t)=0$ for all $t\ge 0$. More generally, if $q_i(0)>0$, then the queue decreases until it reaches zero, and remains zero thereafter.\qed
\end{proof}

\section{Proof of Proposition~\ref{prop:pos_inv}}
\label{app:pos_inv}
\begin{proof}
We verify that the vector field points inward on each boundary component of $\mathcal X$.

\noindent\emph{(i) Nonnegativity ($x \ge 0$).} We analyze the vector field at the boundary $x = 0$ as follows. If $s_i^{(d)}(t)=0$, then:
\[
\dot{s}_i^{(d)}(t)
=
\sum_{k\neq i}\mu_{ki} f_{ki}^{(d)}(t)\ge 0.
\]
Similarly, if $f_{ij}^{(d)}(t)=0$, then:
\[
\dot{f}_{ij}^{(d)}(t)
=
\mu_s p_{ij}^{(d)} s_i^{(d)}(t)\ge 0.
\]
Hence trajectories cannot leave the nonnegative orthant.

\noindent\emph{(ii) Conservation of mass $(\mathbf 1^\top x(t) = 1)$.}
Define the total mass:
\[
M(t):=\mathbf 1^\top x(t)
=
\sum_d\left(\sum_i s_i^{(d)}(t)+\sum_{i\neq j} f_{ij}^{(d)}(t)\right) = 1.
\]
Summing the dynamics over all $i,j,d$ gives:
\begin{align*}
\frac{d}{dt}M(t) = \frac{d}{dt} \left( \sum_d\left(\sum_i s_i^{(d)}(t)+\sum_{i\neq j} f_{ij}^{(d)}(t)\right) \right) =0.
\end{align*}
where the last equality follows from expanding the terms of $\dot{s}_i^{(d)}$  and $\dot{f}^{(d)}_{ij}$ the fact that the flight terms in~\eqref{eq:pos_dyn_1}-\eqref{eq:pos_dyn_2} cancel by relabeling indices, and using the fact that
$\sum_{j\neq i} p_{ij}^{(d)}=1$ for each $i$.
Hence, if $\mathbf 1^\top x(0)=1$, then $\mathbf 1^\top x(t)=1$ for all $t\ge 0$.

Moreover, conservation holds separately for each destination class $d$. Define:
\[
M^{(d)}(t)
:=
\sum_i s_i^{(d)}(t)
+
\sum_{i\neq j} f_{ij}^{(d)}(t).
\]
Summing the dynamics over $i,j$ for fixed $d$ yields:
\begin{align*}
\frac{d}{dt}M^{(d)}(t)
&=
\sum_i \left(\sum_{k\neq i}\mu_{ki} f_{ki}^{(d)} - \mu_s s_i^{(d)}\right) \\
&\quad+
\sum_i\sum_{j\neq i}\left(\mu_s p_{ij}^{(d)} s_i^{(d)} - \mu_{ij} f_{ij}^{(d)}\right).
\end{align*}
Rearranging terms, the flight terms cancel by relabeling indices, and using
$\sum_{j\neq i} p_{ij}^{(d)}=1$ for each $i$, we obtain:
\[
\frac{d}{dt}M^{(d)}(t)=0.
\]
Thus:
\begin{equation}
M^{(d)}(t)=M^{(d)}(0)=m_d, \qquad \forall t\ge 0, \label{eqref: Mass d Conservation}
\end{equation}
where:
\begin{equation}
m_d
:=
\sum_i s_i^{(d)}(0)
+
\sum_{i\neq j} f_{ij}^{(d)}(0). \label{eqref: Initial Condition}
\end{equation}

Hence the mass of each destination class is conserved individually, and the dynamics depend only on the initial class masses $\{m_d\}$, not on their internal distribution.

\noindent\emph{(iii) Capacity constraint $(s_i \le \kappa_i)$.}
Define the aggregate service load:
\[
s_i(t)=\sum_d s_i^{(d)}(t).
\]
Summing the service equations over $d$ yields:
\[
\dot s_i(t)
=
\sum_d\sum_{k\neq i}\mu_{ki} f_{ki}^{(d)}(t)
-\mu_s \sum_d s_i^{(d)}(t)
=
a_i(t)-\mu_s s_i(t).
\]
If $s_i(t)=\kappa_i$, then by the definition of $\mathcal X$:
\[
\dot s_i(t)
=
a_i(t)-\mu_s \kappa_i
< 0.
\]
Hence, we conclude that the vector field points inward on the boundary $s_i=\kappa_i$, so the service-capacity constraint is preserved.

Combining (i)--(iii), trajectories starting in $\mathcal X$ remain in $\mathcal X$ for all $t\ge 0$. Hence $\mathcal X$ is positively invariant. \qed
\end{proof}

\section{Proof of Proposition~\ref{prop:obj_func_derivation}}
\label{app:obj_func_derivation}
\begin{proof}
Fix a destination class \(d\), and let \(Y_n^{(d)}\) denote the embedded jump chain with transition matrix $P^{(d)}$.
Assume the chain is ergodic, with stationary distribution \(\pi^{(d)}\) satisfying:
\begin{equation}
\pi^{(d)} = \pi^{(d)} P^{(d)},
\qquad
\pi_i^{(d)} \geq 0,
\qquad
\sum_i \pi_i^{(d)} = 1.
\label{eq: pi = pi P }
\end{equation}
Under stationarity, we have:
\[
\mathbb{P} (Y_n^{(d)} = i) = \pi_i^{(d)},
\]
and therefore the joint probability of a jump \(i \to j\) is:
\[
\mathbb{P} (Y_n^{(d)} = i,\; Y_{n+1}^{(d)} = j)
=
\pi_i^{(d)} p_{ij}^{(d)}.
\]
Hence the expected one-jump cost is:
\[
\mathbb{E}\!\left[\beta_{Y_n^{(d)},Y_{n+1}^{(d)}}\right]
=
\sum_{i,j} \pi_i^{(d)} p_{ij}^{(d)} \beta_{ij}.
\]

Substituting in~\eqref{eq:obj_function} yields:
\[
J
=
\sum_d m_d J^{(d)}
=
\sum_d m_d \sum_{i,j} \pi_i^{(d)} p_{ij}^{(d)} \beta_{ij}.
\]


Substituting $
\beta_{ij}^{(d)} := \tau_{ij} + h\,\mathbf{1}[j \neq d],
$
we obtain:
\[
J
=
\sum_d m_d \sum_{i,j} \pi_i^{(d)} p_{ij}^{(d)}
\left(\tau_{ij} + h\,\mathbf{1}[j \neq d]\right).
\]
Expanding gives:
\[
J
=
\sum_{d,i,j} m_d\,\pi_i^{(d)} p_{ij}^{(d)} \tau_{ij}
+
h \sum_{d,i,j} m_d\,\pi_i^{(d)} p_{ij}^{(d)} \mathbf{1}[j \neq d].
\]
Using:
\[
\mathbf{1}_{[j\neq d]} = 1 - \mathbf{1}_{[j=d]},
\]
the second term becomes:
\[
h \sum_d m_d
\left(
\sum_{i,j} \pi_i^{(d)} p_{ij}^{(d)}
-
\sum_i \pi_i^{(d)} p_{id}^{(d)}
\right).
\]
Since each \(P^{(d)}\) is row-stochastic and \(\sum_i \pi_i^{(d)}=1\),
\[
\sum_{i,j} \pi_i^{(d)} p_{ij}^{(d)} = 1.
\]
Moreover, from stationarity \(\pi^{(d)}=\pi^{(d)}P^{(d)}\):
\[
\pi_d^{(d)} = \sum_i \pi_i^{(d)} p_{id}^{(d)}.
\]
Hence:
\[
J
=
\sum_{d,i,j} m_d\,\pi_i^{(d)} p_{ij}^{(d)} \tau_{ij}
+
h\sum_d m_d(1-\pi_d^{(d)}).
\]
Since \(\sum_d m_d=1\), this simplifies to:
\[
J
=
\sum_{d,i,j} m_d\,\pi_i^{(d)} p_{ij}^{(d)} \tau_{ij}
+
h
-
h\sum_d m_d \pi_d^{(d)}.
\]
Therefore, up to the additive constant \(h\), minimizing \(J\) is equivalent to minimizing:
\[
\sum_{d,i,j} m_d\,\pi_i^{(d)} p_{ij}^{(d)} \tau_{ij}
-
h\sum_d m_d \pi_d^{(d)}.
\]

\end{proof}

% ...

\end{document}